\documentclass{article}

\usepackage{microtype}
\usepackage{graphicx}
\usepackage[dvipsnames]{xcolor}
\usepackage{subfigure}
\usepackage{booktabs} 
\usepackage{amsmath}
\usepackage{amssymb}
\usepackage{amsfonts}
\usepackage{amsthm}
\usepackage{bbm}
\usepackage{multicol,multirow}
\usepackage{listingsutf8}
\usepackage{iftex}
\usepackage{tikz, pgfplots,pgfplotstable}
\usepackage[utf8]{inputenc}
\usetikzlibrary{positioning}
\usetikzlibrary{backgrounds}
\usetikzlibrary{shapes,arrows}
\usetikzlibrary{decorations.pathmorphing}
\usepackage{tikz-dependency}
\usetikzlibrary{shapes}
\tikzstyle{every picture}+=[remember picture]
\pgfplotstableread[col sep=tab]{alphabet_small.tsv}{\alphabet}
\definecolor{color1}{HTML}{AECEA1}
\definecolor{color2}{HTML}{488389}
\definecolor{color3}{HTML}{383C65}
\definecolor{color4}{HTML}{2B1F3E}

\DeclareUnicodeCharacter{03B1}{$\alpha$}
\DeclareUnicodeCharacter{2211}{$\Sigma$}
\DeclareUnicodeCharacter{3BC}{$\mu$}
\theoremstyle{theorem}
\newtheorem{theorem}{Theorem}
\theoremstyle{definition}

\usepackage{url}
\usepackage{hyperref}

\usepackage[accepted,nohyperref]{icml2021}

\icmltitlerunning{Self-Supervised Pretraining of Graph Neural Networks for the Retrieval of Related Mathematical Expressions}

\begin{document}

\twocolumn[
\icmltitle{Self-Supervised Pretraining of Graph Neural Networks for the Retrieval of Related Mathematical Expressions in Scientific Articles}



\icmlsetsymbol{equal}{*}

\begin{icmlauthorlist}
\icmlauthor{Lukas Pfahler}{tudo}
\icmlauthor{Katharina Morik}{tudo}
\end{icmlauthorlist}

\icmlaffiliation{tudo}{Artificial Intelligence Group, TU Dortmund University, Dortmund, Germany}
\icmlcorrespondingauthor{Lukas Pfahler}{lukas.pfahler@tu-dortmund.de}

\icmlkeywords{Machine Learning, ICML}

\vskip 0.3in
]


\makeatletter
\renewcommand{\ICML@appearing}{This work is licensed under a 
\href{https://creativecommons.org/licenses/by-sa/4.0/}{Creative Commons Attribution-}
\href{https://creativecommons.org/licenses/by-sa/4.0/}{ShareAlike 4.0 International License}
.}
\makeatother
\printAffiliationsAndNotice{}  

\begin{abstract}
Given the increase of publications, search for relevant papers becomes tedious. In particular, search across disciplines or schools of thinking is not supported. This is mainly due to the retrieval with  keyword queries: technical terms differ in different sciences or at different times. Relevant articles might better be identified by their mathematical problem descriptions. Just looking at the equations in a paper already gives a hint to whether the paper is relevant. 
Hence, we propose a new approach for retrieval of mathematical expressions based on machine learning.
We design an unsupervised representation learning task that combines embedding learning with self-supervised learning.
Using graph convolutional neural networks we embed mathematical expression into low-dimensional vector spaces that allow efficient nearest neighbor queries.
To train our models, we collect a huge dataset with over 29 million mathematical expressions from over 900,000 publications published on arXiv.org. The math is converted into an XML format, which we view as graph data.
Our empirical evaluations involving a new dataset of manually annotated search queries show the benefits of using embedding models for mathematical retrieval. 
  
This work was originally published at KDD 2020 \cite{Pfahler/Morik/2020a}.
\end{abstract}

\section{Introduction}
Machine learning has contributed to many success stories of search engines.
Unfortunately, the search is most often based on words or text. Technical terms in different disciplines, however, may have different meanings or the same meaning may be referred to by different terms. 
For instance, various usages of the Bayes' law occur in different scientific fields and can be found under different titles. For instance in astrophysics, it is known as \emph{information field theory} \cite{Ensslin/etal/2008c}. Without knowing physics and even if the name \emph{Bayes} were not mentioned, it is easily recognized by the formula
$P(d|s)=P(d,s) / P(s)$ in the paper. 
Another example is Ising's paper in a physics journal from 1925 under the title \emph{Ferromagnetismus}. Today, the Ising model is also popular in machine learning, but is referred to first as \emph{Hopfield network} and later as \emph{Boltzmann machine}. 
This illustrates the aspect of time: words for particular topics change over time. 
The language of Ising's paper is German, the paper introducing Jensen's inequality in 1906 is written in French. Again, the inequality 
$f((a+b)/2) \leq f(a)/2 + f(b)/2$
 can be easily understood, anyhow.
We conclude that the most compact and comprehensive way to transport the main ideas of scientific manuscripts in disciplines like computer science or physics are the equations used. Thus it should also be the way we formulate our search queries when searching for scientific manuscripts.
In order to judge the relevance of mathematical expressions for a search query, a system has to generalize between different notations and match the parts of equations, that describe the same concepts, even if they appear in a different form. A human reader resorts to domain knowledge acquired over years of training in his field to judge the relevance. We wonder how machine learning models with access to vast amounts of mathematical content can help automatize this process.

In this work, we propose to use graph convolutional neural networks to learn a representation of mathematical expressions that captures semantic relatedness. To this end, we design two unsupervised learning tasks, one classic embedding learning task based on contextual similarity and one self-supervised learning task inspired by masked-language models.
We curate a dataset of over 28.9 million equations from over 900,000 papers from arXiv.org and represent the equations as graphs with one-hot encoded features.
Then we train our models on this large collection of equations.
We compile an evaluation dataset with annotated search queries from several different disciplines and showcase the usefulness of our approach for deploying a search engine for mathematical expressions.

The rest of this paper is structured as follows: We begin by reviewing related work on math search and on our machine learning approaches. In Section~3 we describe the dataset of papers and equations sourced from arXiv.org for our study and present our pre-processing choices. Then we present the graph convolutional neural network we use for embedding equations and describe our two unsupervised learning tasks in Section~4. We begin a statistical analysis of our problem in Section~5 before presenting an extensive empirical validation in Section~6.

\section{Math Search and KDD}
Mining and indexing mathematical expressions in document collections is a challenging task, mostly tackled in the information retrieval community \cite{Guidi/Sacerdoti/2015a,Zhong/Zanibbi/2019a}. 
We outline how the problem of math search is treated with the tools from knowledge discovery and data mining and present related work on the machine learning methods we chose for our approach.
\paragraph{Representation}
The first question we have to consider is how to represent mathematical expression. Choices can be divided into two categories: those for visually representing and those for semantically representing math. The former category is focused on the layout of an expression. The most prominent choices are LaTex, a Turing-complete language used in the publications on arXiv.org, as well as Presentation MathML\footnote{\url{https://www.w3.org/TR/MathML3/}}, an XML dialect for displaying math on the web that we chose in this work.
The latter category includes Content MathML and OpenMath, two similar XML dialects that focus on semantic rather than layout, but also domain-specific languages for symbolic math solvers like Mathematica, that also allow to manipulate and transform formulas.
To the best of our knowledge, no large, public collection of semantic math expressions exists and, unfortunately, converting math from a display-representation, where data is available in large quantities, to a semantic representation which seems more appropriate for searching, is a non-trivial task. Available solutions either use rules and heuristics, e.g. the converter ml2om that translates LaTeX to OpenMath \cite{Timoney/99a}, or also apply machine learning \cite{Wang/etal/2018a}. 
We chose to apply machine learning methods directly on the Presentation MathML representation.
The bottom line of the representation question is that math is expressed in trees, either XML or other parse trees. Our previous work \cite{Pfahler/etal/2019a} may be the notable exception to this: We chose to represent equations as fixed-size bitmaps. While one could argue that this is an unsuitable choice, the multitude of machine-learning or computer-vision approaches that successfully transform images of typeset \cite{Deng/etal/2017a} or hand-written \cite{Munoz/etal/2014a,Mahdavi/etal/2019a} math back to tree-based representations suggests that bitmap representations preserve all required information of tree-based approaches.
\paragraph{Similarity Measure}
The second question is how we compute similarity between formulas. Zanibbi et al. distinguish  text-based, tree-based and spectral approaches \cite{multistage}. Text-based approaches transform tree-structured math into a sequence, for instance by pre-order traversal, and then estimate the similarity using methods for sequences like cosine similarities of bags-of-words or the length of the largest common substring. Tree-based approaches focus on matching trees or subtrees. Typically computing similarities involves solving dynamic-programming problems. Spectral approaches work on paths or partial subtrees in the trees. An example is approach0 \cite{Zhong/Zanibbi/2019a}, that indexes root-leaf paths of operator trees. From matches of the root-leaf paths, they compute the largest common subexpression to score the similarity of two equations. To convert math from LaTeX to the semantic representation of operator trees, the authors use ca. 100 handwritten grammar rules.

A new trend is to use machine learning to learn a similarity measure. A machine learning model maps an equation to a dense, low-dimensional vector. The similarity between these so-called embeddings can be computed via their inner product, which enables fast indexing using a variety of index structures,  including faiss and annoy, designed for efficiently handling millions of these dense, low-dimensional vectors.
Mansouri et al. \yrcite{Mansouri/etal/2019a} propose to embed equations using fastText, a method originally designed for computing word embeddings, while in our previous work \cite{Pfahler/etal/2019a} we compute embeddings with a similar embedding learning task and convolutional neural networks (see Section~4.2).

\paragraph{Graph Convolutional Neural Networks}
In this work we propose an embedding model based on graph convolutional neural networks. Like classic convolutional neural networks for image processing, they compute feature maps based on local neighborhoods. 
While in CNNs, we have features associated with each pixel in the pixel grid and neighborhoods are defined by this grid, in graph CNNs we have features associated with each node of the graph and neighborhoods are defined by the edges in the graph.
We define graph structures  $x=(X,E)$ as a tuple of node-features $X$ and edges $E$. Let $|x|$ denote the number of nodes in $x$. We assume that $X \in \mathbb R^{|x| \times d}$ where $X_i$ are the features of the $i$-th node. A graph CNN maps an input graph to an output with transformed feature vectors in a $d'$-dimensional output space but with identical edge structure.
It is defined by composing different layers. Borrowing the notation of Morris et al. \yrcite{morris2019weisfeiler}, an abstract graph network layer is defined by its output $x'_i$ for the $i$-th node
$$x_i' = \psi\left(x_i,\square_{j \in \mathcal{N}_E(i)} \, \phi\left(x_i, x_j, e_{ij}\right) \right)$$
where $\phi, \psi$ are (sub-)differentiable operators such as linear transformations or multi-layer perceptrons, $\square$ denotes a sub-differentiable, permutation invariant function like sum, mean or max and $\mathcal N_E(i)$ denotes the set of all neighboring nodes of $i$ in the graph with edges $E$. We optionally use information about the edges in the form of vectorial edge-features $e_{ij}$.
As long as all layers in a graph neural network are (sub-)differen\-tiable operations, we can train the network via backpropagation. Efficient software libraries for training models with GPU-support are available, e.g. we use torch-geometric \cite{Fey/Lenssen/2019}.
Graph CNNs have been applied in many contexts, for instance for classifying molecules \cite{duvenaud2015convolutional} or classification and segmentation of point-clouds \cite{wang2019dynamic}. In this work we apply them to learn similarities between mathematical expressions, where we view an XML-representation as graph-structured data.

\paragraph{Self-Supervised Learning}
We further draw influence from a recently proposed class of representation learning tasks called self-supervised learning. Self-supervised learning tasks are unsupervised learning tasks, where parts of the inputs are used to construct proxy tasks. The representations learned in these proxy-tasks can then be used in downstream tasks. For instance, we can rotate images and train a model to predict the rotation angle, as proposed by Gidaris et al. \yrcite{Gidaris/etal/2018}. Using massive amounts of unlabeled data readily available, we can fit models that solve a task like this.

We are particularly interested in masking tasks, where parts of the input are hidden from a model and the model's task is to predict the hidden parts. This was made popular by the BERT model for pretraining natural language representations \cite{Devlin/etal/2018a} and has since then been adopted to other inputs, for instance as pretraining for image classification with convolutional neural networks \cite{Trinh/etal/2019}. We construct a masking task for mathematical expressions and use graph convolutional neural networks to predict the masked parts.

\section{The Data}
We outline how we gather data from arxiv.org and transform them to graph structured data for our graph convolutional neural network.
\subsection{Dataset}
We are working on data obtained from arxiv.org, a service where scientists can upload their manuscripts or pre-prints without reviewing process. 
We have downloaded all the LaTeX sources of publications up to April 2019 from the official bulk data repositories\footnote{\url{https://arxiv.org/help/bulk_data_s3}}. This way we have obtained 934,287 papers. As we can see in Figure~\ref{fig:arxiv_stats}, the large majority of these papers are from disciplines where mathematical expressions are an important part of publications. The most prominent subject areas are astrophysics, condensed-matter physics, computer science, mathematics, and high energy physics.

\begin{figure}
\begin{tikzpicture}{}
\begin{axis}[
     width=\columnwidth,
     height=3.6cm,
    ybar stacked,
    ymin = 0,
	bar width=5pt,
    enlarge x limits=0.1,
    extra y tick style={tickwidth=3pt, grid=major},
    extra y tick labels={},
    xtick pos=left,
    ytick pos=left,
    ylabel={Frequency},
    xlabel={Subject Area},
    symbolic x coords={astro-ph,cond-mat,cs,math,hep-ph,physics,quant-ph,hep-th,gr-qc,stat,nucl-th,hep-ex,q-bio,math-ph,hep-lat,nlin,nucl-ex,q-fin,eess,econ,chao-dyn},
    xtick=data,
    xlabel style={yshift=-3mm},
    ylabel style={yshift=-5mm},
    x tick label style={rotate=45, anchor=east, font=\scriptsize},
    ]
\addplot[ybar,fill=color1] plot coordinates {(astro-ph,181705) (cond-mat,153934) (cs,139105) (math,127655) (hep-ph,71913) (physics,61273) (quant-ph,43465) (hep-th,33754) (gr-qc,21115) (stat,20078) (nucl-th,16613) (hep-ex,12810) (q-bio,9984) (math-ph,9068) (hep-lat,9056) (nlin,8244) (nucl-ex,6189) (q-fin,4586) (eess,2611) (econ,444) (chao-dyn,240)};
\end{axis}
\end{tikzpicture}
\caption{Number of Papers per Subject Areas in our Sample}
\label{fig:arxiv_stats}
\end{figure}
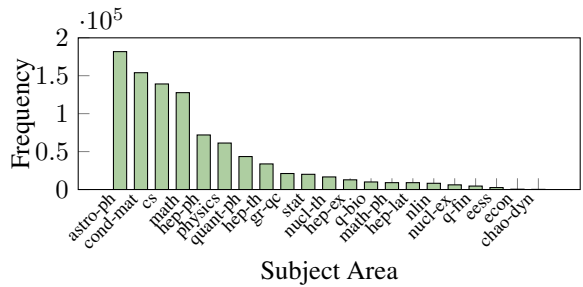

From all publications, we extract mathematical expressions by using regular expressions for the most common math-environments like 'equation', 'align', etc. We do not use inline math snippets but focus on expressions that stand on their own, as they tend to describe more important concepts. Furthermore we extract user-defined commands and macros. Using the library Katex\footnote{\url{http://katex.org}} we compile the raw LaTeX-equations to the XML-based MathML format.
Out of all papers downloaded, 760,041 papers contain at least one equation that we were able to convert to MathML. In total we have a dataset of 28,973,591 MathML equations.
Furthermore we have used regular expressions to find arXiv-ids in the bibliographies of the paper to build a citation graph. In total, 540,892 papers have an outgoing edge, with a total number of edges of 4,553,297. Since we only detect those references that use an arXiv-id, for instance in an url, our citation graph is only a subgraph of the true citation graph.

To ensure reproducibility we provide the scripts used for processing the public arXiv data dump, extracting the mathematical expressions and converting them to MathML as well as collecting meta-data and citations at \url{https://github.com/Whadup/arxiv_library}. We also share our citation graph, which might be interesting in other applications.

\subsection{Data-Representation}
\label{sec:representation}
In order to feed the MathML to a graph convolutional neural network, we have to convert it to a graph with vectorial node features. 
The MathML standard defines around 30 different XML-tags like \texttt{<mi>} for math identifiers or \texttt{<mo>} for math operators. Some of these tags use attributes, for instance to change font or spacing. Leaf nodes contain text like numbers, parenthesis or letters (greek, latin, etc...). 
We view the XML-structure as a tree and use its nodes and edges and derive features based on tags, attributes and text.
For each node we use one-hot encoded feature vectors of dimensionality 256. The first 32 features are used to encode the type of the XML-tag. The next 32 features are used to encode optional attributes, most commonly changes of the font to bold or calligraphy fonts. The remaining 192 dimensions are used to encode the most frequent characters used in leaf-nodes. In Figure~\ref{fig:characters} we see that the most frequent characters are opening and closing parenthesis, followed by a variety of numbers, latin or greek letters and mathematical operators. For both attribute and character features, we introduce special unknown symbols for all rare attributes/characters.
In addition to the one-hot encoded features, we store the position of the node with respect to the parent node.
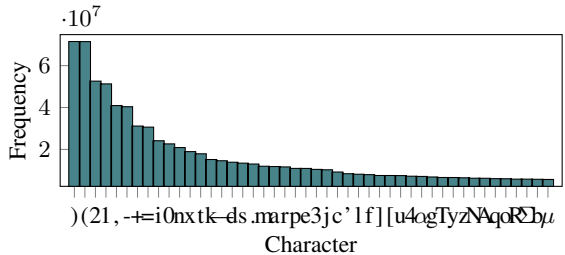
\begin{figure}
\begin{tikzpicture}[tight background, scale=1.00]
      \def\mystrut{\footnotesize\vphantom{())}}
    \pgfplotsset{
    x tick label style={font=\mystrut}
    }
      \begin{axis}[ybar,
        small,
        bar width=1.49mm,
        height=3.6cm,
        width=\columnwidth,
        ylabel={Frequency},
        xlabel={Character},
        xtick pos=left,
        ytick pos=left,
        symbolic x coords={{)},{(},{2},{1},{,},{-},{+},{=},{i},{0},{n},{x},{t},{k},{|},{d},{s},{.},{m},{a},{r},{p},{e},{3},{j},{c},{'},{l},{f},{]},{[},{u},{4},{α},{g},{T},{y},{z},{N},{A},{q},{o},{R},{∑},{b},{μ}},
        enlarge x limits=0.03,
        enlarge y limits=0.05,
        xtick=data,
        ylabel style={yshift=-7mm},
        ]
          \addplot[fill=color2] table [x=char,y=count] {\alphabet};
      \end{axis}
    \end{tikzpicture}
    \caption{The 50 most frequent characters in math environments.}
    \label{fig:characters}
\end{figure}
\section{Learning to Find Related Equations}
In this section we will introduce the graph convolutional neural network used for computing embeddings and present two unsupervised learning tasks used for training the network.
\subsection{Graph-Convolutional Model for Equations}
We define a graph convolutional neural network for the task of embedding mathematical expressions into a low-dimensional vector space.
The raw MathML is converted to graphs with vectorial features as described in Section~\ref{sec:representation}.
We propose to use a special first layer that combines the one-hot encoded information at a node with the decimal position attribute.
Following Vaswani et al. \yrcite{Vaswani/etal/2017a}, we encode the position of the $i$-th node $p_i\in \mathbb N$ using positional embeddings. We use fixed sinusoid embeddings \cite{Vaswani/etal/2017a} denoted by $E(p_j)$, but to still allow the model to control the influence of the positional embeddings, we introduce a learnable scaling coefficient $\alpha$ initialized to 1.
$$x^{(1)}_i = \max\left(\vec 0,\sum_{j \in \mathcal N(i) \cup {i}} W^{(1)}x_j + \alpha E(p_j) + b^{(1)}\right)$$

The first layer is followed by 3 fully-connected graph convolution layers of width 512, where the $l$-th layer is defined by
$$x^{(l)}_i = \max\left(\vec 0,\sum_{j \in \mathcal N(i) \cup {i}} W^{(l)}x^{(l-1)}_j + b^{(l)}\right)$$
which linearly transforms all nodes using a weight matrix $W^{(l)}$, adds a bias term $b^{(l)}$, aggregates by computing the sum over all neighborhoods and applies the ReLU activation component-wise. All graph convolution layers output feature maps with 512 dimensions.
We apply batch-normalization before the first and third graph convolution layer.
For the remainder of this paper, let $\phi(x)\in \mathbb R^{|x| \times 512}$ denote the output of the last graph convolution layer given the input $x$.
To obtain a single embedding for an input graph, we compute the mean of all node features.
This mean is transformed in another linear layer to reduce the dimensionality to 64. For the remainder of this paper, let $\bar \phi(x) \in \mathbb R^{64}$ denote this embedding of $x$.

When scoring similarities between embeddings with margin losses, we need to control the norm of the embeddings, otherwise the notion of adherence to a margin becomes meaningless. 
Among others, Ding et al. \yrcite{Ding/etal/2015a} propose to normalize all embeddings to unit length. We propose a softer normalization inspired by batch normalization \cite{batchnorm} that also allows to obtain embeddings with norms smaller than 1. For every training batch of graphs, we compute the mean of the norm as well as its standard deviation. Then we inversely scale each embedding by the mean plus the standard deviation. This way, most embeddings have norm smaller than 1. We keep a running average of the means and standard deviations. At inference time, we use these running averages for scaling.

\subsection{Representation Learning Tasks}
We propose to train our embeddings using two self-supervised learning tasks simultaneously by adding their respective losses.
\paragraph{Contextual Similarity}
For learning relations between equations, we rely on the established contextual similarity task that was first made popular by word embeddings \cite{Mikolov/etal/2013a} and has hence been used in many representation learning approaches, including our approach \cite{Pfahler/etal/2019a} for learning similarities between equations.
The main idea is that objects that frequently appear in shared contexts are related.
We define the context of mathematical expressions as the paper containing the equation and conjecture that two equations are related if they appear in the same paper, as originally proposed in \cite{Pfahler/etal/2019a}. We extend this approach and further define two equations as related if one paper references the other using a citation graph. This way we hope to connect equations that describe the same context but use different notation.
In addition, we discriminate between sampling expressions from the same paper and from the same section. We hope that within sections, equations are more related to each other.
For obtaining positive examples of related equations, we 
\begin{enumerate}
\item sample a paper uniformly at random and select an expression from this paper uniformly at random.
\item randomly select whether we sample from the same section, same paper or along a citation, \item sample a positive example using that method. When we cannot find a positive example using that method, we jump back to (1).
\end{enumerate}
For learning similarities we also require negative examples. To obtain these, we sample a paper uniformly at random and select an expression from this paper uniformly at random. The random process that generates these weak labels for similarity learning introduces a lot of noise, as many equations we claim are related are unrelated and some of the pairs we say are unrelated are related. We leave the investigation of more advanced sampling schemes to future work.

Using the sampled equations $x$ with positive $x^+$ and negative partners $x^-$, we apply similarity learning. We have to choose a suitable loss function and investigate two  different losses: Triplet and Histogram. 
The triplet loss \cite{Balntas/etal/2016a} we have previously used \cite{Pfahler/etal/2019a}, contrasts the similarity between a positive pair of examples and a negative pair of examples and demands that the similar pair has a higher similarity by a user-defined margin $\Delta$, usually set to 1.

\begin{align}
	\ell_{t}(x,x^+,x^-) = \max (0,\Delta-\langle \bar\phi(x),\bar\phi(x^+)\rangle + \langle \bar\phi(x),\bar\phi(x^-)\rangle) 
\end{align}

In this paper, we propose to use the histogram loss as proposed by Ustinova and Lempitsky \yrcite{Ustinova/Lempitsky/2016a}. It does not work on a triplet of equations, but on a mini-batch of size $m$ positive pairs $X^+$ and a batch of negative pairs $X^-$ with respect to anchor examples $X$.
We collect all similarities between positive pairs in a vector $s^+=(\langle \bar \phi (x_i),\bar\phi (x^+_i)\rangle)_{i=1,...,m}$ and of all negative pairs in $s^-$.
We divide the interval $[-1,1]$ into $R-1$ equally-sized bins with boundaries $-1=t_1,t_2,...,t_R=1$ and width $\Delta=2/(R-1)$ and build histograms for the positive similarities and the negative similarities. Now we demand that the positive histogram leans more toward the $+1$ similarity than the negative histogram. We formalize this intuition as
\begin{equation}
\ell_{h}(s^+, s^-) = \frac 1 {m^2} \sum\limits_{r=1}^R \sum\limits_{r'=1}^r \left(\sum\limits_{i=1}^m \delta_{r}[s^-_i]\right)\left(\sum\limits_{i=1}^m \delta_{r'}[s^+_i]\right) \nonumber
\end{equation}
where instead of hard assignments, we use the triangular kernel
$$\delta_r[s]=\begin{cases}(s-t_{r-1})/\Delta\mbox{ if }s \in [t_{r-1},t_{r}] \\ 
                           (t_{r-1}-s)/\Delta\mbox{ if }s \in [t_{r},t_{r+1}] \\
                          0 \mbox{ otherwise}\end{cases}$$
to put similarities into bins. This way we obtain a differentiable loss function.
We hope that histogram loss is more robust with regard to the massive noise in our labels as each positive example is contrasted with all negative examples.


\paragraph{Masking Task}
We propose to extend the contextual similarity task by another tasks and optimize the sum of both tasks for training our embedding models. The main idea of our second task is, that the symbols in mathematical expressions do not appear independent from each other, but have strong dependencies. Thus if we hide a fraction of the symbols in an equation, we should be able to approximately reconstruct the hidden symbols from the remaining symbols. This task is reminiscent of masked language modeling tasks made popular by BERT \cite{Devlin/etal/2018a} for natural language processing.
In order to successfully solve this task, a model has to learn about the frequencies of symbols and their dependencies from the data, as is illustrated in Figure~\ref{fig:masking}.

\begin{figure}
\begin{equation*}
  \min\limits_{%
  \tikz[baseline]{\node[draw=black,fill=black,line width=0.0pt,rounded corners=.5pt,anchor=base,inner sep=1.5pt] (w1){$\phantom w$};}}
  \frac 1 n \sum\limits_{i=1}^n \ell(\langle w,\tikz[baseline]{\node[draw=black,fill=black,line width=0pt,rounded corners=.5pt,anchor=base,inner sep=1.5pt] (w2){$\phantom w$};}_i\rangle,y_i)
\end{equation*}
\begin{align*}
P( \tikz[baseline]{\node[draw=black,fill=black,line width=0pt,rounded corners=.5pt,anchor=base,inner sep=1.5pt] (p1){$\phantom X$};} = ?) = \begin{pmatrix} w : 0.81 \\ \beta : 0.04 \\ \vdots\end{pmatrix} & & P( \tikz[baseline]{\node[draw=black,fill=black,line width=0pt,rounded corners=.5pt,anchor=base,inner sep=1.5pt] (p2){$\phantom X$};} = ?) = \begin{pmatrix} x : 0.73 \\ y : 0.02 \\ \vdots\end{pmatrix} \\
\end{align*}
\begin{tikzpicture}[overlay]
   \draw[->,line width=.75pt, color=black] (w1.west) to[out=180, in=90] (p1.north);
   \draw[->,line width=.75pt, color=black] (w2.south) to[out=-90, in=90] (p2.north);
\end{tikzpicture}
\caption{Example of the Masking Task with Fictional Values}
\label{fig:masking}
\end{figure}
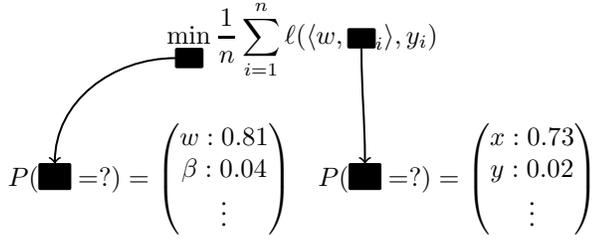

More formally, we proceed as follows: For each input graph $x$ with features $X$, we randomly set the feature vector of $15\%$ of the nodes to all zero obtaining the graph $x_{\blacksquare}$. Then we compute $\phi(x_{\blacksquare})\in \mathbb R^{|x|\times 512}$. Now for each masked node, we want to solve three separate classification tasks: Given $\phi_i(x_{\blacksquare})$, predict the right XML-tag, predict the right XML-attributes (or no-attribute if no attributes where used) and predict the right character (or no-character if no character was used). We solve these tasks using a single linear layer of dimensionality 32,32+1 and 192+1 respectively with soft-max activation and compute the cross-entropy loss $\ell$ on all tasks:
\begin{align*}
\ell_{\text{tag},i} &= \ell(\mathop{\mathrm{softmax}}(W^{(tag)} \phi_i(x_\blacksquare) + b^{(tag)}, X_{i,1:32})\\
\ell_{\text{attr},i} &= \ell(\mathop{\mathrm{softmax}}(W^{(attr)} \phi_i(x_\blacksquare) + b^{(attr)}, X_{i,33:64})\\
\ell_{\text{char},i} &= \ell(\mathop{\mathrm{softmax}}(W^{(char)} \phi_i(x_\blacksquare) + b^{(char)}, X_{i,65:256})\\
\end{align*}
The overall loss of the masking task is defined as the mean of all three classification losses $\ell_{\text{mask},i}= \frac 1 3(\ell_{\text{tag},i}+\ell_{\text{attr},i}+\ell_{\text{char},i})$. The loss is only evaluated for the masked tokens and we compute the mean over all masked tokens to obtain a loss value for $x_\blacksquare$.

Adding this task to the contextual similarity task has the additional advantage that we now learn a representation that not only captures context information, but also preserves information about the raw input symbols.

\subsection{Data-Augmentation}
Data augmentation eases the generalization of machine learning models and is particularly popular for image classification tasks where we can augment images by randomly rotating, scaling, pad\-ding, etc. For mathematical expressions, we propose the following random data augmentation:
Since we know that a renaming of symbols in equations rarely changes the semantic, we propose to randomly permute the character features of all nodes that correspond to a math identifier, encoded in \texttt{<mi>} tags according to the MathML standard.
For each equation we process, we sample a number of flips from a Poisson distribution with expected value $32$. 
Then starting with the identity permutation that does not change the order of our 192 features, we construct a permutation with the desired number of flips by incrementally exchanging two random characters.

\subsection{Hyperparameter Choices}
We train all our models for 20 epochs with Adam optimization, batch size of 128, an initial learning rate of $0.0001$ that we decrease linearly. We use $R=64$ bins for the histogram loss and margin $\delta=1$ for the triplet loss.

\section{Statistical Analysis}
Before we present empirical results on our embedding approach, we want to discuss their statistical significance using concentration inequalities. We begin by discussing the assumptions we use for our discussion that go beyond the usual iid. assumptions. Then we talk about testing of embedding models on hold-out data.
\subsection{Assumptions}
In usual classification settings \cite{Vapnik/2000a}, as well as metric learning settings \cite{Clemencon/etal/2016a,Maurer/2008c}, we assume that we have access to independent and identically distributed training data. In our case of embedding learning for mathematical expressions, we would require to a sample of independent equations, or to be more precise, independent pairs of equations with (weak) similarity labels.
However, this assumption surely does not hold, as two equations that appear in the same paper do not appear independently from each other. An example illustrates this: Let $X_1$ and $X_2$ denote two equations. If it is revealed to you that $X_1$ shows Heisenberg's uncertainty principle $X_1=\left\{\sigma_x\sigma_p \geq\hbar/{2}\right\}$, and that $X_2$ appears in the same paper, the probability of $X_2$ being related to quantum mechanics increases. This illustrates that $P(X_2=x_2 \mid X_1=x_1) \not = P(X_2=x_2)$ and thus $X_1$ and $X_2$ are not independent. 
For the purpose of our analysis, we assume that  
\begin{enumerate}
\item[\textbf{(A1)}] we have an iid. sample of $N$ papers where the $i$-th paper contains $n_i$ equations
\item[\textbf{(A2)}] the number of equations is $n=\sum_{i=1}^N n_i$
\item[\textbf{(A3)}] equations in the $i$-th paper are independent of equations in all other papers.
\end{enumerate}
Surely papers do not appear independently of one another, but this we will ignore in our statistical analysis.
\subsection{Confidence Intervals for Hold-Out Data}
Now we carry out our statistical analysis for measuring scores on hold-out data where the model $\bar\phi$ is fixed and independent of the hold-out data.
In the common iid setting, we can apply concentration inequalities like Hoeffding's inequality to bound the gap between performance evaluations like accuracy or loss measured on a finite hold-out sample and the true expected value.

Since we do not have iid. data, we have to use a refined approach. We use $U$-statistics\cite{Hoeffding/1948a} $s$ to analyze our models. In our case, they are defined as expectations over functions of triples of equations $V(x_1,x_2,x_3)$ called kernel.
$$s = \mathop{\mathbb E}\limits_{x_1,x_2,x_3} V(x_1, x_2, x_3).$$ 
For instance we can define the ranking score, i.e. the fraction of triplets where the positive pair has a higher similarity than the negative pair. For convenience, we define $P(x)$ as the set of possible positive examples for an expression $x$. Then we can write the ranking score as
\begin{equation}
  \begin{split} 
V_{\text{ranking}}(x_1,x_2,x_3) = \Big[\mathbbm 1(x_2\in P(x_1) \wedge x_3 \not \in P(x_1)) \;\cdot \\ \mathbbm 1(\langle \bar\phi(x_1),\bar\phi(x_2)\rangle > \langle \bar\phi(x_1),\bar\phi(x_3)\rangle)
\Big]\end{split}\nonumber
\end{equation}
The histogram loss can also be expressed as a $U$-statistic:
\begin{align*}
V_{\text{hist}}(x_1,x_2,x_3)= \Big[\mathbbm 1(x_2\in P(x_1) \wedge x_3 \not\in P(x_1)) \;\cdot \\\nonumber
\sum\limits_{r=1}^R\sum\limits_{r'=1 }^r \delta_{r'}(\langle\bar\phi(x_1),\bar\phi(x_2)\rangle) \delta_{r}(\langle \bar\phi(x_1),\bar\phi(x_3)\rangle)
\Big]\end{align*}
Note that whenever $(x_1,x_2,x_3)$ is not a triple with a positive and a negative example, the triple contributes zero to the expectations.
To simplify the analysis, we first consider so-called complete $U$-statistics, i.e. we estimate the true score using all possible triplets from a finite set of $n$ expressions $\{x_i \mid {i=1,...,n}\}$
\begin{equation}
\widehat{s}= {\binom n 3}^{-1} \sum\limits_{i,j,k} V(x_i, x_j, x_k). 
\label{eq:empirical_s}
\end{equation}

Now the goal is to bound the difference between $s$ and $\hat s$ in high probability.

\begin{theorem}
If $V(x_1,x_2,x_3) \in [0,1]$, then for $\delta \in (0,1)$ with probability at least $1-\delta$ we have
$$s \leq \hat s + 3\sqrt{\frac{\ln(1/\delta)} {2N}}$$
\end{theorem}
\begin{proof}
We proof this using Janson's concentration inequality for sums of partly dependent variables \cite{Janson/2003}. In Equation (\ref{eq:empirical_s}), we compute a sum over triplets, where the triplets are constructed from dependent samples. Thus the summands have dependencies. We construct a dependency graph, where each node corresponds to a triplet and two nodes have an edge if they are dependent. In our case two triplets are connected if one of the equations in the first triplet is from the same paper as any equation in the second triplet. 
In order to apply \cite{Janson/2003}, Theorem 2.1, we have to bound the chromatic number $\chi^*$ of this dependency graph. To this end, we consider an arbitrary node $v$ in the graph with equations from paper $i$, $j$ and $k$. Its degree is bounded by the number of equations in the papers $\mbox{deg}(v) \leq (n_i+ n_j + n_k) \binom {n-1}2$, hence
\begin{equation}
\chi^* \leq 3\binom {n-1}2 \max_i n_i.
\label{eq:chi}
\end{equation}
Then for $t>0$ we have 
\begin{align}
\mathbb P(s - \widehat s \geq t)
&\overset{\text{\cite{Janson/2003}}}\leq \exp\left(\frac{ -2 t^2 {\binom n 3}}{\chi^*}\right) \\
&\overset{\text{(\ref{eq:chi})}}\leq \exp\left(\frac{ -2 t^2 n}{9 \max_i n_i}\right)
\leq \exp\left(\frac{ -2 t^2 N}{9}\right)\nonumber
\end{align}
Solving for $t$ yields the desired confidence interval.
\end{proof}
For incomplete $U$-statistics, where we estimate the score by a subset of triplets $D$ sampled independently with replacement\cite{Clemencon/etal/2016a} as
$$\widehat s_\subset = \frac 1 {|B|} \cdot \sum\limits_{(x_i,x_j,x_k) \in D} V(x_i, x_j, x_k),$$
the bound (\ref{eq:chi}) no longer holds. But we can compute an empirical bound $\hat\chi$ using any greedy graph coloring algorithm. Then Janson's concentration inequality implies that with probability at least $1-\delta$  
\begin{equation}
  s \leq \widehat s_\subset + \sqrt{\frac{\ln(1/\delta) \hat\chi} {2|D|}}.
\end{equation}



\section{Experimental Results}
In this section we perform an experimental evaluation of our embedding model. In particular, we focus on the use-case of a search engine for mathematical expressions. 
We begin by investigating the effects of the individual components of our model on a small, closed subset of the data.
Then we investigate the effectiveness of our method on all 29,9 million equations. 

\subsection{Analysis on the Machine-Learning-Subset}

We begin our analysis only on arXiv publications where the primary subject classification is machine learning (\texttt{cs.LG}).
This is a natural choice, as we have some expertise to judge the quality of our results, a task which we are in no way equipped for across all subject fields.

Of these 9,936 publications, we sample two subsets, train and test of size 7,949 and 1,987 respectively with a total number of equations of 237,335 and 54,767 respectively. In Section~5 we have seen that scores evaluated on hold-out data converge to true empirical scores in $\mathcal O(1/{\sqrt{N}})$ where $N$ is the number of papers in the test set. In this regard, our test set is appropriately sized.
We use the train-set for building our embedding models and use the test-set to investigate generalization properties.

For training, we sample 1 million triplets $(x,x^+,x^-)$. Of these triples, $45.9\%$ have a positive pair from the same section, $42.2\%$ from the same paper and $13.9\%$ along an edge in the citation graph.
We sample 100k triplets for testing with similarly distributed positive examples.

We perform an ablation study on our proposed embedding model and compare it to prior work.
This section investigates the influence of our design choices. We decided (a) to use the Histogram loss instead of the triplet loss, (b) to also add an masking task, (c) to data augmentation. 

We measure Ranking score, i.e. the fraction of all triples in the training data where same-class pairs of equations have higher similarities than across-class pairs.
\begin{table}[]
\caption{Ablation Study}
\label{table:ableation}
\setlength\tabcolsep{2pt}
\small
\begin{tabular}{@{}lccc@{}}
\toprule
Influence Factor  & \begin{tabular}[c]{@{}c@{}}Ranking\\Hold-Out\end{tabular} & \begin{tabular}[c]{@{}c@{}}Ranking\\Eval\end{tabular} & \begin{tabular}[c]{@{}c@{}}Accuracy\\Eval\end{tabular}  \\ \midrule
Full Model        & 76.5 ($\pm 0.0$)   & 57.7 ($\pm 0.0$) &  60.6 ($\pm 0.0$) \\
No Histogram Loss & 72.5 ($ -4.0$)     & 49.6 ($  - 8.1$) &  30.9 ($  -29.7$) \\
No Masking        & 75.2 ($ -1.3$)     & 54.3 ($  - 3.4$) &  50.0 ($  -10.6$) \\
No Augmentation   & 75.3 ($ -1.2$)     & 53.6 ($  - 4.1$) &  50.0 ($  -10.6$) \\\midrule
Bitmap CNN original\cite{Pfahler/etal/2019a} 
                  & 76.2 ($ -0.3$)     & 71.9 ($ + 14.2$) &  68.3 ($ +  7.7$)        \\ 
Bitmap CNN retrained 
                  & 70.0 ($ -6.5$)     & 50.0 ($  - 7.7$) &  52.9 ($  -7.7$)        \\
\bottomrule
\end{tabular}
\end{table}
As we see in Table~\ref{table:ableation}, our evaluations indicate that all of our design choices contribute favorably to the overall performance on hold-out data, as deactivating any component decreases the score. We note that the biggest gain is achieved by switching form triplet-loss to histogram-loss. We believe that this is due to the massive noise in our labels.

We also compare with the our previous model \cite{Pfahler/etal/2019a} and see that we beat this baseline by a small margin. However this comparison is not entirely fair, as their model was trained on a larger dataset of around 25k papers, probably including some of the papers in our test set. We use their code to re-train on our subset of equations and yield a substantial margin of 6.5 percentage points.

We also use our previous evaluation data \cite{Pfahler/etal/2019a}. It consists of 103 equations labeled into 13 categories related to machine learning including k-means, LSTMs, empirical risk minimization, etc. 
Since only bitmaps are available, we transcribe the equations manually. There are three issues with this evaluation set: First, it is too small to produce significant numbers. 
Second, some equations in the dataset appear in the training data. This is not only the case for our subset, but also for the training data used in \cite{Pfahler/etal/2019a}.
Third, many equations within a category are obviously from the same paper, hence we have seen some of the pairs in our training data. 
Nevertheless we use the evaluation data. Indeed in our use-case of search engines, the crawled equations will always be in the training data and only the user queries will be unseen equations. In a way, we simulate this with the eval data.

Following the original experimental protocol, we measure the 1-nearest-neighbor accuracy obtained in leave-one-out validation (named Accuracy) as well as the above Ranking score. In Table~\ref{table:ableation}, we again see that our model is only surpassed by the pre-trained model that uses a larger training dataset. 
This motivates the use of a much larger dataset.


\subsection{Large-Scale Experiments}
For training on all the papers in our dataset, we sample two different sets of training triplets, one with 5 million triplets and one with 20 million triplets.
We train our models on a Nvidia GTX1080 GPU with 8GB memory, which allows us to process mini-batches of 128 triplets or 384 equations. During training, we process around 1,300 triplets per second, not counting the time for reading data from hard disk. In total, one of the 20 epochs of training on 20mio triplets takes 6:30h on our system. We use annoy to construct an index for approximate nearest neighbor retrieval. In total, our index uses 13GB of hard disk storage to manage all mathematical expressions in our dataset.


\begin{table}[]
\caption{Eval Scores}
\label{table:eval_large}
\centering
\begin{tabular}{@{}lccc@{}}
\toprule
Dataset  & \begin{tabular}[c]{@{}c@{}}Ranking\\Eval\end{tabular} & \begin{tabular}[c]{@{}c@{}}Accuracy\\Eval\end{tabular}  \\ \midrule
1mio ML-Subset Triplets                      & 57.7  &  60.6 \\
5mio Full ArXiv Triplets                     & \textbf{76.2}  &  80.9 \\
20mio Full ArXiv Triplets                    & 75.3 & \textbf{84.0}\\
Bitmap CNN original\cite{Pfahler/etal/2019a}  & 71.9  &  68.3 \\ 
\bottomrule
\end{tabular}
\end{table}
Before we evaluate our models in a search engine study, we again check the performance on the aforementioned evaluation data.
The results in Table~\ref{table:eval_large} indicate the power of using large amounts of training data, although it is unclear if using 20mio training triplets is an advantage over using only 5mio. Our large-scale models beat all the models trained on smaller amounts of data. Even though the smaller models were trained on only machine-learning related data, we obtain better scores on the machine learning evaluation data by training on all disciplines.
\begin{figure}[b]
\centering
\begin{tabular}{lc}
\toprule
Query:& $P(d \mid s ) = \frac{P(d,s)}{P(s)}$\\
1st Result:&$P(s \mid d) = \frac{P(d\mid s) P(s)}{P(d)}$\\
4th Result:&$P(d) = \int P(d\mid s) P(s) ds$\\
\bottomrule
\end{tabular}
\caption{Example: Bayes' law. We report the first result and the first result that does not show Bayes' law, but, in this case, the related law of total probability. The first result is from: R. H. Leike, T. A. Enßlin, \textit{Charting nearby dust clouds using Gaia data only}, 2019.}
\label{fig:example1}
\end{figure}


\begin{figure}[t]
\centering
\small
\begin{tabular}{lc}
\toprule
Query: &$\sum_{i<j} w_{ij} \, s_i \, s_j + \sum_i \theta_i \, s_i$\\
'Boltzmann' Result: & $ E = -\sum_i b_i \, s_i  - \sum_{i<j} w_{ij}  s_i  s_j$. \\
'Ising' Result: & $\mathcal{H}=-\sum_{i<j}C_{ij}\,J_{ij}\sigma_{i}\sigma_{j}-\sum_{i}h_i\sigma_{i}$ \\
\bottomrule
\end{tabular}
\caption{Example: Ising Model. We find equations related to both Ising Models and Boltzmann Machines. First result is from: Weinstein, \textit{Learning the Einstein-Podolsky-Rosen correlations on a Restricted Boltzmann
Machine}, 2017. Second result is from: Ferrari et al., \textit{Finite size corrections to disordered systems on Erdös-Rényi random graphs}, 2013.}
\label{fig:example2}
\end{figure}
Let us now inspect two example search queries. In Figures~\ref{fig:example1} and \ref{fig:example2} we see two examples from the introduction, Bayes law and Ising models, and their respective nearest neighbors under our model trained on 5mio triplets. We see that we can find other definitions of Bayes' law as well as the related law of total probability. When we perform a query for the Ising Model and look at the first 20 results, we find papers where the model is called Boltzmann machine as well as papers that refer to the Ising model. This illustrates the power of querying for mathematical expressions instead of using keywords.

\subsection{Search Engine Study}
Finally we want to study the usefulness of our embedding approach for a search engine application more systematically. 
Traditionally, validating search engines using measures like precision or recall, requires relevance scores for each result for each evaluation query. We see that this requires a lot of manual annotation work since we have to manually identify each relevant equation for each query.
Unfortunately, we were not able to find available evaluation data. The best fit is the NTCIR-12 task evaluation data \cite{Zanibbi/etal/2016a} consisting of 37 annotated queries. However is not appropriate for our approach, as most queries are a combination of math as well as keywords. When we ignore the keywords, the remaining query become very generic, for instance $x+y$, which makes it very unlikely that we accurately find the articles labeled as relevant. In addition, the overall focus of the NTCIR-12 task is recovery of exact matches, whereas our focus is on retrieving \emph{related} expressions.

Consequently, we curate and publish our own evaluation data. To reduce the manual annotation labour, we want to apply a heuristic for the relevance judgement. To this end, we have asked our colleagues, many from disciplines other than computer science and data science, to provide us with equations that we should query. For each equation, they provide a set of keywords or key-phrases that should appear in the section around the result. If one of the keywords is present, we count the result as correct. This way we can evaluate our search result without manually checking result lists. If a keyword has more than 10 characters, we also count it, if we find a substring that has a Levenshtein distance less than 2.
In total, we have 53 evaluation queries publicly available and editable online\footnote{Crowd-sourced evaluation data can be accessed and edited here: \url{https://www.overleaf.com/8721648589nrjxgwmtzfvm}.}.

We inspect two different information retrieval metrics that do not require to know the number of relevant documents in advance: Precision$@k$ and unnormalized Mean Average Precision.
Precision$@k$ is defined as the fraction of relevant documents within the first $k$ results. We report it for lists of 10, 100 and 1000 results and compute its mean over our evaluation queries.

Unnormalized Mean Average Precision is derived of the standard mean average precision metric. Since we do not now the number of relevant documents in advance, we omit this term, limit the search to a maximum of 1000 results and obtain the following definition
$$\textrm{uMAP} = \sum\limits_{k=1}^{1000} P(k) \Delta_k$$
where $P(k)$ is Precision$@k$ and $\Delta_k$ specifies if the $k$-th result is relevant. Again we compute the mean over all evaluation queries. In comparison to Precision$@k$, uMAP considers the order of the search results and rewards relevant results early in the result lists.

For reference, we include retrieval based on a bag-of-words representation. To this end, we use our data representation as in Section~\ref{sec:representation}, but compute the sum over all nodes in the graph to obtain a single 256-dimensional vector of the whole tree. We retrieve the nearest neighbors using cosine similarity.

In Table~\ref{table:precisions}, we see that our approach beats the bag-of-words margin, in particular for larger values of $k$.
We see for Precision$@10$, the performance between BoW and our embedding model is very similar. This is because for many queries the top-10 results are mostly near-perfect matches which are easily identified even without machine learning. However when we look at more results, we are able to find almost 50\% more relevant equations.

\begin{table}[]
\caption{Search Engine Performance}
\label{table:precisions}
\centering
\begin{tabular}{@{}lrrrr@{}}
\toprule
      & P$@$10 & P$@$100 & P$@$1000 & uMAP \\ \midrule
BoW   & 0.4567 & 0.3170 & 0.2083 & 106.17  \\
5Mio  & \textbf{0.5038} & \textbf{0.3817} & \textbf{0.2984} & \textbf{165.04} \\
20Mio & 0.4547 & 0.3709 & 0.2897 & 156.51       \\ \hline
\end{tabular}
\end{table}

Overall the precision values seem very low. This is in part due to the experiment design where we rely on the annotated keywords. A closer inspection reveals that the queries achieve very different precision values.
In Table.~\ref{table:sorted_p100} we show the best and worst-performing queries of our 5Mio training examples model. There are examples that achieve more than 90$\%$ precision, but many queries have precision lower than $1\%$. We note that the results are better if keywords are broader. Highly specific queries where the number of relevant documents is low perform poorly. We hope that in the future our collection of evaluation query grows further to allow more informative evaluations. In particular with more queries it may be helpful to split the evaluation data into difficulty levels according to the number of relevant documents.

\begin{table}[]
\caption{Sorted P$@100$ values per Query for 5Mio Triplets}
\label{table:sorted_p100}
\small
\centering
\begin{tabular}{@{}rl@{}}
\toprule
P$@100$ & Keywords                                                                                               \\ \midrule
0.95 & 'policy gradient' \\
0.93 & 'convex', 'strongly convex'\\
0.91 & 'q-learning', 'reinforcement learning'\\
0.89 & 'chain complex', 'sequence'\\
0.81 & 'lipschitz', 'continuous'\\
0.80 & 'empirical risk', 'training objective', 'erm'\\
0.71 & 'lstm', 'recurrent neural network'\\
0.70 & 'neural network', 'hidden layer', ...\\
0.67 & 'eigenvalue', 'eigenvector',...\\
0.66 & 'received signal strength indication', 'rssi', ...\\
 & $\vdots$ \\
0.08 & 'jaccard similarity', 'min-hashing', ...\\
0.07 & 'fredholm integral', 'equation of the first kind', ...\\
0.07 & 'effective collection area', 'effective area'\\
0.07 & 'modified erlang', 'erlang'\\
0.07 & 'johnson-lindenstrauss', 'embedding'\\
0.06 & 'significance of detection', 'lima'\\
0.06 & 'log-odds alignment', 'pairwise dynamic prog.', ...\\
0.05 & 'proximal gradient'\\
0.04 & 'handshaking lemma', 'handshake lemma'\\
0.03 & 'single photon spectrum', 'multi-gaussian distr.', ...\\
0.02 & 'graphcnn', 'geometric deep learning', ...\\
0.01 & 'crosstalk probability'\\ \bottomrule
\end{tabular}
\end{table}

\subsection{Automatic Identification of Equalities}

\begin{table*}[t]
	\centering
    \setlength{\tabcolsep}{0.25em}
	\caption{Results of the Mathematical Retrieval Experiment. We report recall@K for $K\in\{1,10,100\}$.}
	\label{table:retrieval2}
	\begin{tabular}{l c  c  c  c  c  c  c  c  c }
		\toprule
		{\textbf{Model}} & \multicolumn{3}{c}{\textbf{Equalities (36864)}} & \multicolumn{3}{c}{\textbf{Relations (40960)}} & \multicolumn{3}{c}{\textbf{Inequalities (13312)}}\\
         & {R@1} & {R@10} & {R@100} & {R@1} & {R@10} & {R@100} & {R@1} & {R@10} & {R@100}
		\\ \midrule
		Transformer Model & \bf 0.511 & 0.71 & 0.805 & 0.503 & 0.697 & 0.791 & 0.484 & 0.765 & 0.86
		\\ 

		BoW & 0.483	& 0.653 & 0.739 & 0.491 & 0.658	& 0.743 & 0.503	& 0.738	& 0.821 \\ 
		FastText \cite{Joulin/etal/2017a} & 0.480 & 0.650 & 0.739 & 0.488 & 0.651	& 0.742 & 0.488	& 0.713	& 0.810 \\ 
		Graph-CNN  & 0.507 & \bf 0.833 & \bf0.884 & \bf0.512 & \bf 0.834 & \bf 0.883 & \bf 0.504 & \bf 0.870 & \bf 0.922\\
		\bottomrule
	\end{tabular}

\end{table*}

We have extracted equalities and inequalities in the test-set of our data using regular expressions. Using a simple heuristic, we filter the resulting (in-)equalities, such that left-hand-side (LHS) and right-hand-side (RHS) do not differ in length dramatically, thereby eliminating formulas like definitions, where the LHS is a only single symbol. We derive three different data sets, one with only equalities (LHS and RHS split at "$=$"), one with inequalities (split at $<$ and $\leq$) and one with mixed relations (split at $=<>\leq$ and $\geq$).
This data allows us to use the LHS of the (in-)equalities as query in hopes of retrieving RHS. We make our finetuning-data available at \url{https://whadup.github.io/arxiv_learning/} as well.

Following other machine-learning-based approaches for mathematical retrieval \cite{Mansouri/etal/2019a,Pfahler/etal/2019a,Pfahler/Morik/2020a}, we use our models to encode formulas into a dense vector space and retrieve results using approximate nearest neighbor search \cite{Bernhardsson/2018a} in this vector space. 
We finetune our models on half of the available data and test on the remaining half.
\paragraph{Finetuning Task} We propose to use contrastive learning to learn to identify the RHS given the LHS.
The learning task in contrastive learning is identifying the right partner for each input in a minibatch of datapoints. Hence the representation learning problem is formulated as a classification problem.
Let $X^l, X^r  \in \mathbb R^{m \times d}$ contain the output embeddings of a minibatch of LHSs and RHSs.
We normalize each embedding to unit-length and denote the normalized embeddings by $\bar X^l$ and $\bar X^r$.
We use the InfoNCE loss\cite{Oord/etal/2018a}, i.e. the negative log likelihood of softmax probabilities parameterized by the pairwise cosine similarities between the LHs and RHSs:

\begin{equation}
    \ell_\tau(\bar X^l, \bar X^r) = m^{-1} \sum_{i=1}^m \log \frac
    {\exp(\langle \bar X^l_i, \bar X^r_i\rangle / \tau)}
    {\sum_{j\not =i} \exp(\langle \bar X^l_i, \bar X^r_j\rangle / \tau)}
\end{equation}
where $\tau > 0$ is a hyperparameter that controls the temperature of the ouput probability distribution, which we set to $10^{-2}$.
The contrastive learning task is more difficult for larger batchsizes $m$, as there are more candidate RHSs to choose from and thus the underlying classification problem becomes more difficult. But it has been shown that the utility of the model increases for larger batchsizes \cite{Chen/etal/2020a,Misra/Maaten/2020a}, which we also investigate in our application.

\paragraph{Baseline-Models}
In addition to our models we include several baseline models:
\begin{itemize}
\item 
First, we test a simple \emph{bag-of-words} (BoW) model that is trained on a bag of MathML tree nodes. This model does not use a pre-training phase, but is only tuned on the finetuning data. 
The BoW model maps the sparse BoW representation to a $d$-dimensional vectorial embedding though a single matrix multiplication. 
In comparison to our BERT models, we do not restrict the vocabulary size of the inputs. 
The representation is trained using the same contrastive learning task with InfoNCE loss. We test $d\in\{64,128,256\}$ and report the best result after varying learning rates and number of training epochs in a grid search.
\item
Second we evaluate a pretraining approach based on the BoW-model. The word-embedding based approach 
\emph{fastText} by Joulin et al.\cite{Joulin/etal/2017a} is trained by predicting which tokenss appear in the contexts together. Mansoury et al. use it for learning embeddings of formulae by serializing a MathML layout tree similar to the one we are using in this work \cite{Mansouri/etal/2019a}, hence we include it in our comparison. We finetune these embeddings by learning a linear mapping into a $d$-dimensional vector space, $d\in\{64,128,256\}$ using the same contrastive learning task.
\item Third, we also train a BERT \cite{Devlin/etal/2018a} transformer model on our mathematical expressions using the same two pretraining tasks: masked token prediction and same-paper prediction. This model has 6 million trainable parameters and uses 8 transformer layers with a hidden dimensionality of 256 and an intermediate dimensionality of 768. We use the same vocabulary as for our graph neural network.
\end{itemize}
We begin by training our models and the baseline-models with a minibatch-size of 1024. Then we also investigate the effect of varying the batch-size.
\paragraph{Results}
For testing, we compute embeddings for all LHSs and RHSs in the test-data and store them in an index structure.
We use annoy \cite{Bernhardsson/2018a}, an indexing method for approximate nearest-neighbor search based on an ensemble of random projection trees. We use an ensemble of 16 trees with default hyperparameters, but we found that the results were very insensitive to our particular parameter choices. 

Then we query the $k$-nearest neighbors, $k\in \{1,10,100\}$ for each formula from the test-set and check if the corresponding other side of the (in-)equality is in the result set. This way we can compute recall values to measure the quality of our embeddings.

As we can see in Table~\ref{table:retrieval2}, our Graph-CNN network substantially outperforms both baseline methods -- bag-of-words and FastText -- but also the transformer model.
With one exception, we can report the best recall scores on all 3 subsets of formulas, often by quite substantial margins.

\section{Conclusion and Outlook}
Finding relevant literature even across disciplines is essential for scientific investigations. The search results should entail stimulating, relevant papers. 
Very often, a look at the formula in a paper gives a compact description of the problems and solutions discussed in the paper. Hence, the goal is to offer related papers based on the mathematical expressions. This task is different from mathematical information retrieval, but it shares the problem of determining the right representation of mathematical expressions. 

In this paper, we have proposed and evaluated a new method for searching mathematical expressions based on machine learning. The problem is framed as representation learning on graph-structured data. A precise definition of this task with its assumptions is given. Using concentration inequalities for sums of partly dependent variables allowed us to analyze the performance of our embedding models with statistical assumptions that go beyond the usual i.i.d. over-simplifications. Further work could extend this into a PAC style analysis.

For the first time, we have applied unsupervised embedding learning with graph convolutional neural networks to learn a representation of math that allows efficient retrieval of semantically related expressions. Unlike existing work, our approach does not rely on hand-written rules, but learns embeddings purely data-driven in a combination of two unsupervised tasks: On the one hand, we train a contextual similarity tasks where labels are generated from the surrounding contexts of mathematical expressions, on the other hand we train a self-supervised masking task where labels are derived directly from the inputs.
To illustrate our ideas, we have curated a huge dataset with over 29 million mathematical expressions based on over 900,000 papers hosted on arXiv.org. This allowed us to train graph convolutional neural networks with millions of equations and to carefully examine the impact of our design choices.
We have shown the benefit of our search system using a new dataset of annotated math queries.

As of now, our method uses uniform sampling for positive and negative examples. In the future we want to explore more guided variants of sampling. It is still unclear what the tradeoffs between less noise in the labels and more bias in the sampling are.



%
Given the importance of the data samples and their labeling, we have set up a crowd-sourced evaluation procedure.
Currently, there are more evaluation queries available than in the NTCIR dataset. This data collection will be continued and extended to user interactions with a running search engine. 
\subsection*{Acknowledgements}
We thank our fellow scientist who contributed annotated equations to our evaluation dataset. 
This work has been supported by Deutsche Forschungsgemeinschaft (DFG) within the Collaborative Research Center SFB 876, "Providing Information by Resource-Constrained Data Analysis", project A1. \url{http://sfb876.tu-dortmund.de}. Parts of this work have been funded by the Federal Ministry of Education and Research of Germany as part of the competence center for machine learning ML2R (01IS18038A).
\ifPDFTeX
\bibliographystyle{icml2021}
\else
\renewcommand{\refname}{References \hfill\normalfont{\color{Plum}[fix the bibstyle before submission]}}
\bibliographystyle{apalike}
\fi
\bibliography{main.bib}

\end{document}